\documentclass[a4paper,12pt]{article}
\usepackage{amsmath,amsthm,amssymb,bbm}

\usepackage[english]{babel}

\newtheorem{corollary}{Corollary}
\newtheorem{lemma}{Lemma}
\newtheorem{definition}{Definition}
\newtheorem{proposition}{Proposition}

\DeclareMathOperator{\e}{e}

\newcommand{\kb}[2]{\left\vert#1\right\rangle\left\langle#2\right\vert}

\newcommand{\ket}[1]{\left\vert #1 \right>}

\newcommand{\modu}[1]{\left| #1 \right|}
\newcommand{\bkmv}[3]{\left< #1 \right\vert#2 \left\vert #3 \right>}

\title{\bf Entanglement in fermion systems and quantum metrology}

\author{F. Benatti$^{a,b}$, 
R. Floreanini$^{b}$, U. Marzolino$^{c}$\\
\\
\small ${}^a$Dipartimento di Fisica, Universit\`a di Trieste, 
34151 Trieste, Italy\\
\small ${}^b$Istituto Nazionale di Fisica Nucleare, Sezione di Trieste,
34151 Trieste, Italy\\
\small ${}^c$Department of Physics, University of Ljubljana, 1000 Ljubljana, Slovenia}

\date{\null}

\begin{document}

\maketitle

\begin{abstract}
\noindent
Entanglement in fermion many-body systems is studied using 
a generalized definition of separability based on partitions of the 
set of observables, rather than on particle tensor products. In this way,
the characterizing properties of non-separable fermion states can be explicitly analyzed,
allowing a precise description of the geometric structure of the corresponding state space.
These results have direct applications in fermion quantum metrology:
sub-shot noise accuracy in parameter estimation can be obtained without the need of a preliminary
state entangling operation.

\end{abstract}

\section{Introduction}

In trying to apply the standard definition of separability and entanglement to systems of identical
particles, one immediately faces a problem: the indistinguishability of the system constituents
conflicts with Hilbert space tensor product structures on which these notions are based.
The point is that the particles are identical and therefore they can be neither singly addressed, 
nor their individual properties measured:
only collective, global system operators are in fact admissible, 
experimentally accessible observables \cite{Feynman,Sakurai}.

This observation unavoidably leads to a radical change in perspective concerning the attitude towards the notion of
entanglement in general: the presence of quantum correlations in any physical system is less signaled
by {\it a priori} properties of the system states, 
than by those of the algebra of the system observables 
and by the behaviour of the associated correlation functions. In other terms, the usually adopted 
definition of separability based on the particle aspect of first quantization appears to be too restrictive,
leading possibly to misleading results; rather, it should be replaced by one directly emerging from 
the second quantized description, usually adopted for studying many-body systems.%
\footnote{Entanglement
in many-body systems has been widely discussed in the recent literature, {\it e.g.} see
\cite{Schliemann}-\cite{Modi}; however, for the reasons just pointed out,
only a limited part of those results are really applicable to the case of identical particle systems.}

This new approach to separability and entanglement has been advocated before \cite{Zanardi}-\cite{Viola2}, 
but formalized only recently \cite{Benatti1}-\cite{Marzolino}. 
So far the focus has been on bosonic systems, with particular attention to bipartite entanglement,
aiming at specific applications to quantum metrology. 
Suitable criteria able to detect non-classical
correlations in systems with a fixed number of elementary bosonic constituents
have been discussed. In particular, 
it has been found that in general 
the operation of partial transposition \cite{Peres,Horodecki2} gives rise to a much more exhaustive
criterion for detecting bipartite entanglement than in the case of 
distinguishable particles \cite{Argentieri,Benatti3}. 
This allows obtaining a rather complete classification
of the structure of bipartite entangled states in systems composed by a fixed
number of bosons \cite{Benatti3,Benatti4}.
Further, entangled bosonic states turn out to be much more robust 
than distinguishable particle ones against
mixing with other states and an explicit expression for the so-called
``robustness'' \cite{Vidal,Steiner} has been derived \cite{Benatti4}.
In this way, a general characterization of the geometry of the space of
bosonic states can be given, that is indeed much richer than in the case of systems
of distinguishable constituents.

In the following, we shall extend the study of the notions of separability and entanglement
to the case of systems composed of fermions following the lines previously adopted for bosons. In this case,
the elementary creation and annihilation operators associated
with the fermion system constituents satisfy an algebra given in terms
of anti-commutation relations: this poses
new questions regarding the connection between the properties of locality and commutativity
of the system observables, making the theory of fermion entanglement even richer
than in the case of bosonic systems. 

Application to quantum metrology using fermion systems will also be discussed:
as in the case of bosonic systems, also in the case of fermion ones it will be
explicitly shown that sub-shot noise accuracy in parameter estimation can be
achieved without the need of feeding the measuring apparatus with entangled states;
the required non-locality can be provided by the apparatus itself.
These results may have practical implications in interferometric experiments
using ultracold fermion gases.

\section{Entanglement in multimode fermion systems}

We shall consider generic fermion many-body systems made of $N$
elementary constituents that can occupy $M$ different orthogonal states or modes, $N<M$. This is a quite general
model that can accommodate various physical situations in atomic and condensed matter physics;
in particular, it can be used to describe the behavior of ultracold gases
confined in multi-site optical lattices, that are becoming so relevant in the study of
quantum many-body phenomena ({\it e.g.}, see
\cite{Shi}-\cite{Amico}, \cite{Leggett1}-\cite{Yukalov} and references therein).

A many-body system made of identical particles is usually described
by means of creation and annihilation 
operators, $a_i^\dagger$, $a_i$, for each of the $M$ modes that the particles can occupy,
$i=1, 2,\ldots,M$ \cite{Strocchi,Thirring}.
Fermion particles are characterized by the fact that the operators $a_i^\dagger$, $a_i$,
obey the canonical anti-commutation relations, 
\begin{equation}
\{a_i,\,a^\dagger_j\}\equiv a_i\,a^\dagger_j+a^\dagger_j\, a_i=\delta_{ij}\ ,\quad\{a_i,\,a_j\}=
\{a_i^\dagger,\,a^\dagger_j\}=\, 0\ .
\label{1}
\end{equation}
The total Hilbert space $\cal H$ of the system
is then spanned by the many-body Fock states, obtained by applying creation operators to the
vacuum:
\begin{equation}
|n_1, n_2,\ldots,n_M\rangle= 
(a_1^\dagger)^{n_1}\, (a_2^\dagger)^{n_2}\, \cdots\, (a_M^\dagger)^{n_M}\,|0\rangle\ ,
\label{2}
\end{equation}
the integers $n_1, n_2, \ldots, n_M$ representing the occupation numbers of the different modes;
due to (\ref{1}), they can take only the two values 0 or 1.  
Since the number of fermions $N$ is fixed,
the total number operator $\sum_{i=1}^M a_i^\dagger a_i$
is a conserved quantity and the occupation numbers must satisfy the 
additional constraint $\sum_{i=1}^M n_i=N$; in other words, all states must
contain exactly $N$ particles and the dimension $D$ of the system Hilbert space $\cal H$
is then: $D={M\choose N}$.
In addition, the set of polynomials in all creation and annihilation operators,
$\{a^\dagger_i,\, a_i\, |\,i=1,2,\ldots, M\}$,
form an algebra that, together with its norm-closure, coincides with the algebra
${\cal A}({\cal H})$ of bounded operators acting on $\cal H$; the observables of the systems are part of this algebra.

As mentioned in the introductory remarks, in this framework the natural interpretation of entanglement in terms
of particle correlations has to be rethought. For instance, 
in the case of a system composed by two standard, distinguishable qubits, the natural Hilbert space product 
structure ${\cal H}=\mathbb{C}^2\otimes\mathbb{C}^2$ and 
the corresponding algebraic product structure for the space of the associated observables
${\cal A}=M_2(\mathbb{C})\otimes M_2(\mathbb{C})$, with $M_2(\mathbb{C})$ the set of $2\times2$ complex matrices, immediately identify the local observables as the one taking the form
\begin{equation}
\label{3}
A\otimes B=(A\otimes 1)\,(1\otimes B)\ ,
\end{equation}
where $A$ is an observable of the first qubit, while $B$ that for the second one.
In other terms, local observables for the two-qubit systems are characterized by being tensor products of observables pertaining each to one of the two parties: they commute and are thus algebraically independent.

Consider instead a system composed by two fermions that can occupy two modes, and thus described
by the set of operators $(a_1, a_1^\dagger, a_2, a_2^\dagger)$: 
the single particle Hilbert space is still $\mathbb{C}^2$; the difference with respect to the qubit case is that 
the total Hilbert space $\cal H$ is now one-dimensional, containing just one Fock vector, namely:
$a^\dag_1a^\dag_2\vert 0\rangle$. In the language of first quantization, this correspond to the fact
that only anti-symmetric states are allowed due to the Fermi statistics, and this is automatically
enforced in the second quantized language due to the algebra in (\ref{1}).
Further, the algebra $\cal A$ of operators is linearly generated by the identity 
together with at most second order monomials in $a_1,a^\dag_1$ and $a_2,a_2^\dag$.
In this case, the particle Hilbert space tensor product structure is lost,
reflecting the fact that the two particles are indistinguishable.
Similarly, also the usual notion of local observables, the one
based on the tensor product structure as in (\ref{3}), is no longer available and 
need to be reformulated.

In dealing with systems of identical particles,
it is natural to define the notion
of bipartite entanglement by the presence of non-classical correlations among 
averages of operators.
It is then convenient to start with the following general definitions, 
valid for both boson and fermion systems:

\begin{definition}
An {\bf algebraic bipartition} of the operator algebra ${\cal A}({\cal H})$ is any pair
$({\cal A}_1, {\cal A}_2)$ of subalgebras of ${\cal A}({\cal H})$ generated by disjoint
subsets of modes, namely
${\cal A}_1, {\cal A}_2\subset {\cal A}({\cal H})$, ${\cal A}_1 \cap {\cal A}_2={\bf 1}$.
\end{definition}

\noindent
In general the two subalgebras ${\cal A}_1$ and ${\cal A}_2$ need not reproduce the whole
algebra ${\cal A}({\cal H})$, {\it i.e.} 
\hbox{${\cal A}_1 \cup {\cal A}_2\subset {\cal A}({\cal H})$};
however, in the cases of partitions defined in terms of modes, as discussed below,
one has: ${\cal A}_1 \cup {\cal A}_2={\cal A}({\cal H})$.

Any algebraic bipartition encodes in a natural way the definition of the system local observables:
 
\medskip
\noindent
\begin{definition}
An element (operator) of ${\cal A}({\cal H})$ is said to be 
$({\cal A}_1, {\cal A}_2)$-{\bf local}, {\it i.e.} local with respect to
a given bipartition $({\cal A}_1, {\cal A}_2)$, if it is the product $A_1 A_2$ of an element 
$A_1$ of ${\cal A}_1$ and another $A_2$ in ${\cal A}_2$.
\end{definition}

From this notion of operator locality, a natural definition of state separability and entanglement
follows \cite{Benatti1}:

\medskip
\noindent
\begin{definition}
A state $\omega$ on the algebra ${\cal A}({\cal H})$ will be called {\bf separable} with
respect to the bipartition $({\cal A}_1, {\cal A}_2)$ if the expectation $\omega(A_1 A_2)$ 
of any local operator $A_1 A_2$ can be decomposed into a linear convex combination of
products of expectations:
\begin{equation}
\omega(A_1 A_2)=\sum_k\lambda_k\, \omega_k^{(1)}(A_1)\, \omega_k^{(2)}(A_2)\ ,\qquad
\lambda_k\geq0\ ,\quad \sum_k\lambda_k=1\ ,
\label{4}
\end{equation}
where $\omega_k^{(1)}$ and $\omega_k^{(2)}$ are given states on ${\cal A}({\cal H})$;
otherwise the state $\omega$ is said to be {\bf entangled} with respect the bipartition
$({\cal A}_1, {\cal A}_2)$.%
\footnote{In general, a state $\omega$ is a normalized, positive, linear functional on ${\cal A}({\cal H})$, such that
the average of any observable $\cal O$ can be expressed as the value taken by $\omega$ on it,
$\langle {\cal O}\rangle=\omega({\cal O})$; a standard representation of this expectation value map
is given by the trace operation over density matrices.}
\end{definition}

\noindent
{\bf Remark 1:} {\sl i)} This generalized definition of separability 
can be easily extended to the case
of more than two partitions, by an appropriate, straightforward generalization; 
specifically, in the case of an $n$-partition, Eq.(\ref{4})
would extend to:
\begin{equation}
\omega(A_1 A_2\cdots A_n)=\sum_k\lambda_k\, \omega_k^{(1)}(A_1)\, \omega_k^{(2)}(A_2)\cdots
\omega_k^{(n)}(A_n)\,\ ,\quad
\lambda_k\geq0\ ,\quad \sum_k\lambda_k=1\ .
\label{5}
\end{equation}

\noindent
{\sl ii)} As already observed before, in systems of identical particles
there is no {\it a priori} given, natural partition to be used for the definition
of separability; therefore,
issues about entanglement and non-locality 
are meaningful {\it only} with reference to a choice of a specific partition in the operator algebra
\cite{Zanardi}-\cite{Benatti5}; 
this general observation, often overlooked, is at the origin of
much confusion in the recent literature.\hfill$\Box$

\smallskip

\noindent

In the language of second quantization introduced before for the description of fermion systems, 
these general definitions can be made more explicit.
A bipartition
of the $M$-modes fermion algebra ${\cal A}({\cal H})$ can be given by splitting the collection
of creation and annihilation operators into two disjoint sets,
$\{a_i^\dagger,\, a_i\, | i=1,2\ldots,m\}$ and 
$\{a_j^\dagger,\, a_j,\, |\, j=m+1,m+2,\ldots,M\}$; 
it is thus uniquely determined by
the choice of the integer $m$, with $0\leq m \leq M$.%
\footnote{There is no loss of generality 
in assuming the modes forming the two partitions to be contiguous; 
if in the chosen bipartition this is not the case, 
one can always re-label the modes in such a way to achieve this convenient ordering.}

All polynomials in the first set (together with their norm-closures)
form a subalgebra ${\cal A}_1$, while the remaining set analogously generates
a subalgebra ${\cal A}_2$. 
Due to the anti-commutation relations (\ref{1}), the two sub-algebras ${\cal A}_1$,
${\cal A}_2$ do not in general commute. Nevertheless, from the algebraic relations
\begin{equation}
\nonumber
[AB\,,\,C]\,=\,A\,\{B\,,\,C\}-\{A\,,\,C\}\,B\ ,
\end{equation}
it follows that all even powers of elements in $\mathcal{A}_1$ ($\mathcal{A}_2$) 
commute with all elements of $\mathcal{A}_2$ ($\mathcal{A}_1$).
It is then convenient to introduce the following definition:

\begin{definition}
Let $\Theta$ be the automorphism on the Fermi algebra ${\cal A}$
defined by $\Theta(a_i)=-a_i$, $\Theta(a_i^\dagger)=-a_i^\dagger$ 
for all $a_i,\ a^\dagger_i\in{\cal A}({\cal H})$.%
\footnote{In other terms, $\Theta$ is a linear map on ${\cal A}$ 
preserving the algebra relations, {\it i.e.} $\Theta(AB)=\Theta(A)\Theta(B)$.}
The even component $\mathcal{A}^e$ of $\mathcal{A}$ is the subset of elements $A^e\in\mathcal{A}$ 
such that $\Theta(A^e)=A^e$, while the odd component $\mathcal{A}^o$ of $\mathcal{A}$ consists 
of those elements $A^o\in\mathcal{A}$ such that $\Theta(A^o)=-A^o$.
\end{definition}

\noindent
Notice that the even component $\mathcal{A}^e$ is the algebra generated by even polynomials 
in creation and annihilation operators, while the odd component $\mathcal{A}^o$ is just a linear space, 
but not an algebra, since the product of two odd elements is even.

Similarly, given the algebraic bipartition $(\mathcal{A}_1,\, \mathcal{A}_2)$,
one can define the even $\mathcal{A}_i^e$ and odd $\mathcal{A}_i^o$ components of the two
subalgebras $\mathcal{A}_i$, $i=1,2$. Only the operators of the first partition 
belonging to the even component $\mathcal{A}_1^e$
commute with any operator of the second partition and, similarly, only the even operators
of the second partition commute with the whole subalgebra $\mathcal{A}_1$.

Coming now back to the notion of separability introduced in {\sl Definition 3},
one may notice that there is a difference between bosonic and fermionic systems.
In the bosonic case, the two subalgebras 
${\cal A}_1$, ${\cal A}_2$ defining the algebraic bipartition $({\cal A}_1,\, {\cal A}_2)$
naturally commute, {\it i.e.} that each element $A_1$ of the operator algebra ${\cal A}_1$ 
commutes with any element $A_2$ in ${\cal A}_2$.
Instead, in the case of fermion systems, as already observed 
the two subalgebras ${\cal A}_1$, ${\cal A}_2$ do not in general commute.
Nevertheless, in such systems only selfadjoint operators belonging to the even components $A_1^e$, $A_2^e$ 
qualify as physical observables and these do commute.

At this point, two different attitudes are possible regarding the definition of separability
expressed by the condition (\ref{4}): {\it i)} use in it all operators from the two subalgebras
${\cal A}_1$, ${\cal A}_2$, as assumed in {\sl Definition 3} above; {\it ii)} restrict all
considerations to observables only. The first approach is in line with the notion of
``microcausality'' adopted in {\sl constructive} quantum field theory \cite{Streater,Strocchi2},
where the emphasis is on quantum fields, which are required either to commute (boson fields)
or anticommute (fermion fields) if defined on (causally) disjoint regions. 
On the other hand, the second point of view reminds
of the notion of ``local commutativity'' in {\sl algebraic} quantum field theory \cite{Emch,Haag},
where only observables are considered, assumed to commute if localized in disjoint regions.

These two points of view are not equivalent, as it can be appreciated by the following
simple example. Let us consider the system consisting of just one fermion that can occupy
two modes, {\it i.e.} $N=1$, $M=2$, with the bipartition defined by the two modes. 
The following state:
\begin{equation}
\omega=|\phi\rangle\langle \phi|\ ,\qquad 
|\phi\rangle=\frac{|1,0\rangle + |0,1\rangle}{\sqrt2}\ ,
\label{6}
\end{equation}
combination of the two Fock basis states $|1,0\rangle$, $|0,1\rangle$ introduced in (\ref{2}),
appears to be entangled.%
\footnote{As we shall see in Section 4, a $N$-fermion generalization of this state can be used
in quantum metrology to achieve sub shot-noise accuracy in parameter estimation.}
Nevertheless,
in the second approach mentioned above, it is found to satisfy the condition (\ref{4}), hence to be separable. 
Indeed, only observables, {\it i.e.} selfadjoint, even operators,
can be used in this case as $A_1$ and $A_2$; in practice, only the two partial number operators
$a_1^\dagger a_1$ and $a_2^\dagger a_2$ together with the identity are admissible, and for these observables
the state (\ref{6}) behaves as the separable state $(|0,1\rangle\langle 0,1| + |1,0\rangle\langle 1,0|)/2$.
Different is the situation within the first approach: in this case, all operators
are admissible, for instance $A_1=a_1^\dagger$ and $A_2=a_2$,
which indeed prevent the separability condition (\ref{4}) to be satisfied.

In view of this, here we advocate and adopt the first point of view, {\it i.e.} point {\it i)} above: 
it gives a more general and physically
complete treatment of fermion entanglement.
Nevertheless, it should be stressed that
the fermion algebra put stringent constraints
on the form of the fermion states that can be represented as product of other states,
as the ones appearing in the decomposition (\ref{4}). Specifically, any product
$\omega_k^{(1)}(A_1)\, \omega_k^{(2)}(A_2)$ vanishes 
whenever $A_1$ and $A_2$ both belong to the odd components of their respective subalgebras. 
This fact comes from the following result \cite{Araki}, whose proof we explicitly give for it is a
direct illustration of the effects of the anti-commutative character of the fermion algebra:
\begin{lemma}
Consider a bipartition $(\mathcal{A}_{1},\mathcal{A}_{2})$ of the fermion algebra $\mathcal{A}$ 
and two states $\omega_1$, $\omega_2$ on $\mathcal{A}$. Then, the linear 
functional $\omega$ on $\mathcal{A}$ defined by 
$\omega(A_1A_2)=\omega_1(A_1)\,\omega_2(A_2)$ for all $A_1\in\mathcal{A}_1$ and $A_2\in\mathcal{A}_2$ is a 
state on $\mathcal{A}$ only if at least one $\omega_i$ vanishes on the odd component of $\mathcal{A}_i$.
\end{lemma}

\begin{proof}
Suppose both states $\omega_1$ and $\omega_2$ do not vanish when acting on the odd components $\mathcal{A}^o_{1,2}$.
Then, there exist odd elements $A^o_i\in\mathcal{A}_i^o$, such that $\omega_i(A^o_i)\neq 0$, $i=1,2$.
The same is true for the self-adjoint combinations $(A^o_i+(A^o_i)^\dag)/2$ and $(A^o_i-(A^o_i)^\dag)/(2i)$: 
we can then assume $(A^o_i)^\dag=A^o_i$, so that
$\overline{\omega_i(A^o_i)}=\omega_i(A^o_i)\neq 0$, where the overline signifies complex conjugation. 
But then, due to the anti-commutativity of the odd elements $A^o_i$, one finds:

$$ 
\overline{\omega(A^o_1A^o_2)}= \omega(A^o_2A^o_1)=-\omega(A^o_1A^o_2)=\omega_1(A^o_1)\,\omega(A^o_2)\neq 0\ ,
$$ 
which is a contradiction.
\end{proof}

In other terms, given a mode bipartition $({\cal A}_1,{\cal A}_2)$ of the fermion algebra $\cal A$,
{\it i.e.} a decomposition of $\cal A$ in the subalgebra ${\cal A}_1$ generated by the first 
$m$ modes and the subalgebra ${\cal A}_2$, generated by the
remaining $M-m$ ones, the decomposition (\ref{4}) is meaningful only for local operators
$A_1 A_2$ for which $[A_1,\, A_2]=\,0$, so that, at this stage, the definition 
of separability that it encodes appears similar to the resulting one for bosonic systems.

As a further consequence of {\sl Lemma 1}, the following criterion for entanglement holds:%
\footnote{This criterion precisely detects the entanglement of the state $|\Phi\rangle$ in (\ref{6});
indeed, with the odd elements $A_1^o=a_1^\dagger$ and $A_2^o=a_2$,
one has $\langle\Phi| A_1^o A_2^o|\Phi\rangle=1/2$.}

\begin{corollary}
Given the bipartition $(\mathcal{A}_{1},\mathcal{A}_{1})$ of the fermion algebra $\mathcal{A}$,
if a state $\omega$ is non vanishing on a local operator $A_1^o A_2^o$, with the two components 
$A_1^o \in {\cal A}_1^o$, $A_2^o \in {\cal A}_2^o$ both belonging to the odd part 
of the two subalgebras, then $\omega$ is entangled.
\end{corollary}

\noindent
Indeed, if $\omega(A_1^o A_2^o)\neq0$, then, by {\sl Lemma 1}, $\omega$ can not be written as
in (\ref{4}), and therefore it is non-separable.

The case of pure states, {\it i.e.} states that can not be written as a convex
combination of other states, deserves a separate discussion. Indeed, in this case
the separability condition (\ref{4}) simplify and the following result can be proven:

\begin{lemma}
Pure states $\omega$ on the fermion algebra $\cal A$ are separable with respect to a given
bipartition $({\cal A}_1,{\cal A}_2)$ if and only if
\begin{equation}
\omega(A_1 A_2)=\omega(A_1)\, \omega(A_2)\ ,
\label{7}
\end{equation}
for all local operators $A_1 A_2$.
\end{lemma}

\begin{proof}
The {\it if} part of the proof is trivial: according to {\sl Definition 3}, states as above are manifestly
$({\cal A}_1,{\cal A}_2)$-separable, since they are of the form (\ref{4}) with just one element in the
the convex sum.

For the {\it only if} part of the proof, recall that, in the case of mode bipartition we are considering, 
the $({\cal A}_1,{\cal A}_2)$-local operators generate the whole fermion algebra $\cal A$.
Therefore, any element $A\in{\cal A}$ can be written as a combination of local operators,
$A=\sum_{ij} C_{ij}\, A_1^{(i)}\, A_2^{(j)}$, with $A_1^{(i)}\in {\cal A}_1$ and 
$A_2^{(j)}\in {\cal A}_2$. As a consequence, if by hypothesis a state $\omega$ is separable,
{\it i.e.} it can be written as in (\ref{4}) on all $({\cal A}_1,{\cal A}_2)$-local operators,
then one has:
$$
\omega(A)=\sum_{ij} C_{ij}\sum_k \lambda_k \ \omega_k^{(1)}(A_1^{(i)})\ \omega_k^{(2)}(A_2^{(j)})=
\sum_k \lambda_k\, \Omega_k(A)\ ,
$$
in terms of other states, defined on the whole algebra $\cal A$ through the relation:
$\Omega_k(A)=\sum_{ij} C_{ij}\, \omega_k^{(1)}(A_1^{(i)})\, \omega_k^{(2)}(A_2^{(j)})$.%
\footnote{Notice that, by {\sl Lemma 1}, $\Omega_k(A_1 A_2)$ is different from zero only when $A_1$ and $A_2$
are not both odd elements.}
But since $\omega$ is pure by hypothesis, only one term in the above convex combination 
must be different from zero. By dropping the now superfluous label $k$, we have then found that:
$\omega(A_1 A_2)=\omega^{(1)}(A_1)\, \omega^{(2)}(A_2)$. The final form (\ref{7}) is obtained
by separately taking $A_1$ and $A_2$ to coincide with the identity operator.
\end{proof}

\section{Structure of entangled fermion states}

The above discussion shows that, for many-body systems formed by $N$ fermions that can occupy $M$-modes, 
the notion of entanglement can not be given once for all, but needs to be referred 
to the choice of a partition
of the modes into two disjoint sets, the first containing the first $m$ modes,
while the second the remaining $M-m$ ones. In short, we will henceforth refer to such a choice 
as the $(m,\ M-m)$-partition. It turns out that, once the partition $(m,\ M-m)$ is fixed,
the general structure of entangled $N$-fermion states can be explicitly described.

Let us first consider the case of pure states; their complete characterization is
given by the following

\begin{proposition}
A pure state $\ket{\psi}$ in the fermion Hilbert space $\mathcal{H}$
is $(m,\ M-m)$-separable if and only if it is generated out of the vacuum state 
by a $(m,\ M-m)$-local operator, {\it i.e.}
it can be written in the form
\begin{equation}
\label{8}
\ket{\psi} = \mathcal{P}(a^\dagger_1,\ldots,a^\dagger_m) \cdot 
\mathcal{Q}(a^\dagger_{m+1},\ldots,a^\dagger_M) \ket{0}\ ,
\end{equation}
where $\mathcal{P}$, $\mathcal{Q}$ are polynomials in the creation operators relative 
to the first $m$ modes and the last $M-m$ modes, respectively. Otherwise, the state is entangled.
\end{proposition}

\begin{proof}
First of all, recall that in the present situation the condition of
separability reduces to the simpler expression (\ref{7}); clearly, the state
in (\ref{8}) satisfies it by taking for the expectation value of a generic
fermion operator $A\in {\cal A}$
the usual state-average: $\omega(A)\equiv\langle\psi| A|\psi\rangle$.

In order to prove the converse, {\it i.e.}
that from the separability condition (\ref{7}) the expression (\ref{8}) follows,
we start decomposing $|\psi\rangle$ in the Fock basis given in (\ref{2});
taking into account the $(m, M-m)$-bipartition of the modes, one can write
\begin{equation}
|\psi\rangle=\sum_{\{k\},\{\alpha\}} C_{\{k\},\{\alpha\}}\ | k_1, \ldots, k_m\, 
;\alpha_{m+1}, \ldots, \alpha_M\rangle\ ,\qquad 
\sum_{ \{k\}, \{\alpha\} } \left| C_{\{k\},\{\alpha\}}\right|^2=1\ ,
\label{9}
\end{equation}
where $\{k\}=(k_1,k_2,\ldots,k_m)$, 
respectively $\{\alpha\}=(\alpha_{m+1},\alpha_{m+2},\ldots,\alpha_{M})$,
is the vector of occupation numbers of the first $m$, respectively second $M-m$ modes, and
$$
\ket{k_{1},\dots,k_{m};\alpha_{m+1},\dots,\alpha_{M}}=
(\hat{a}_1^\dag)^{k_1}\cdots(\hat{a}_m^\dag)^{k_m}(\hat{a}_{m+1}^\dag)^{\alpha_{m+1}}\cdots
(\hat{a}_M^\dag)^{\alpha_M}\ket{0}\ .
$$
The condition of separability of {\sl Lemma 2} assures that
$\langle\Psi\vert A_1A_2\vert\Psi\rangle=\langle\Psi\vert A_1\vert\Psi\rangle\,\langle\Psi\vert A_2\vert\Psi\rangle$
for all fermion operators $A_1$, $A_2$ belonging to the first, second partition, respectively.%
\footnote{Note again that, due to {\sl Lemma 1}, at least one of the two operators $A_1$, $A_2$
need to be even.}
We will use this request to force the coefficients $C_{\{k\},\{\alpha\}}$ to be in product form,
$C_{\{k\},\{\alpha\}}=C_{\{k\}}\,C'_{\{\alpha\}}$, through suitable choice of $A_1$ and $A_2$.
To this aim, let us consider the following operators
\begin{eqnarray}
\label{aid2}
&&A_{1} = ({a}^\dag_{1})^{p'_{1}}\dots({a}^\dag_{m})^{p'_{m}}\,
\left(\dfrac{1}{2\pi i} \oint_{\mit\Gamma} \dfrac{dz}{z-{N}_{1}}\right)\, 
{a}_{m}^{p_{m}}\dots {a}_{1}^{p_{1}}\ , \\
\label{aid3}
&&A_{2} = ({a}^\dag_{m+1})^{\beta'_{m+1}}\dots({a}^\dag_{M})^{\beta'_{M}}\,
\left(\dfrac{1}{2\pi i} \oint_{\mit\Gamma} \dfrac{dz}{z-{N}_{2}}\right)\, 
{a}_{M}^{\beta_{M}}\dots {a}_{m+1}^{\beta_{m+1}}\ ,\\
\nonumber
&&A_1 A_2 =({a}^\dag_{1})^{p'_{1}}\dots({a}^\dag_{m})^{p'_{m}}({a}^\dag_{m+1})^{\beta'_{m+1}}\dots({a}^\dag_{M})^{\beta'_{M}}\\
&&\hskip 5cm \times \left(\dfrac{1}{2\pi i} \oint_{\mit\Gamma} \dfrac{dz}{z-{N}}\right)\,
{a}_{M}^{\beta_{M}}\dots {a}_{m+1}^{\beta_{m+1}}{a}_{m}^{p_{m}}\dots {a}_{1}^{p_{1}}\ ,
\end{eqnarray}
where $p_i$, $p'_i$, $\beta_j$, $\beta'_j$ are either 0 or 1, 
${N}_1=\sum_{k=1}^m{a}^\dag_k\,{a}_k$, ${N}_2=\sum_{j=m+1}^M {a}_j^\dag\,{a}_j$ 
are the number operators relative to the two sub-sets of modes, while $N=N_1+N_2$ is the total number
of fermions in the system; further,
$\mit \Gamma$ is a contour around $z=0$ excluding all other integers. 
The choice of contour forces the three integrals above to vanish unless $z=0$, 
whence the first two project onto the sub-spaces with no particles in the first, second partition, respectively,
while the third one onto the vacuum.%
\footnote{Since the number operators are sums of quadratic monomials, by series expansion 
the three integrals provide operators that are elements of even subalgebras.}
With a slight abuse of language, the three operators above can be represented in short as 
$A_1=\kb{\{p\}}{\{p'\}}$, $A_2=\kb{\{\beta\}}{\{\beta'\}}$ and
$A_1 A_2=\kb{\{p\},\{\beta\}}{\{p'\},\{\beta'\}}$.
Then, using (\ref{9}), one easily obtains
\begin{eqnarray}
\label{13}
&&\bkmv{\Psi}{A_1}{\Psi} =
\sum_{\{\alpha\}} \overline{C}_{\{p'\},\{\alpha\}} C_{\{p\},\{\alpha\}}\ ,\\
\label{14}
&&\bkmv{\Psi}{A_2}{\Psi} =
\sum_{\{k\}} \overline{C}_{\{k\},\{\beta'\}} C_{\{k\},\{\beta\}}\ ,\\
\label{15}
&&\bkmv{\Psi}{A_1A_2}{\Psi} =\overline{C}_{\{p'\},\{\beta'\}} C_{\{p\},\{\beta\}}\ ,
\end{eqnarray}
and the assumed separability of $|\psi\rangle$ yields the condition:
\begin{equation}
\label{16}
\overline{C}_{\{p'\},\{\beta'\}} C_{\{p\},\{\beta\}}  =  
\left( \sum_{\{\alpha\}} \overline{C}_{\{p'\},\{\alpha\}} C_{\{p\},\{\alpha\}} \right) 
\left( \sum_{\{k\}} \overline{C}_{\{k\},\{\beta'\}} C_{\{k\},\{\beta\}} \right)\ .
\end{equation}
For $p'=p$ and $\beta'=\beta$ this expression becomes
$$
\modu{C_{\{p\},\{\beta\}}}^2 =  \left( \sum_{\{\alpha\}} \modu{C_{\{p\},\{\alpha\}}}^2 \right)
\left( \sum_{\{k\}} \modu{C_{\{k\},\{\beta\}}}^2 \right)\ .
$$
Setting $D_{\{ p \}}=\sum_{\{ \alpha \}} \modu{C_{\{p\},\{\alpha\}}}^2$ and $D'_{\{ \beta \}}=
\sum_{\{ k \}} \modu{C_{\{k\},\{\beta\}}}^2$, one can rewrite
\begin{equation}
\label{17}
C_{\{p\},\{\beta\}} = \sqrt{D_{\{p\}} \phantom{'}}\, \sqrt{D'_{\{\beta\}}}\, \e^{i \theta_{\{p\}\{\beta\}}}\ .
\end{equation}
Inserting this expression in \eqref{16}, we obtain
$$
\e^{i(\theta_{\{p'\}\{\beta'\})}-\theta_{\{p\}\{\beta\}})}= \sum_{\{ \alpha \}} D'_{\{\alpha\}} \e^{i \left( \theta_{\{p\}\{\alpha\}}-\theta_{\{p'\}\{\alpha\}}\right)}\times \sum_{\{k\}} D_{\{k\}}
\e^{i \left( \theta_{\{k\}\{\beta\}}-\theta_{\{k\}\{\beta'\}}\right)}\ .
$$
Since due to the state normalization condition $\sum_{\{p\}} D_{\{ p \}}=1=\sum_{\{\beta\}} D'_{\{ \beta \}}$, 
by setting $\beta'=\beta$
one sees that $\theta_{\{p\}\{\beta\}}- \theta_{\{p'\}\{\beta\}}=\phi_{pp'}$ for all $\beta$,
{\it i.e.} this phase difference is a function of the set of indices $p$ and $p'$, but not of $\beta$.
Fixing an arbitrary $p'$ and inserting this expression into \eqref{17} yields
\begin{equation}
C_{\{p\},\{\beta\}} = \sqrt{D_{\{p\}}\phantom{'}}\,\e^{i \phi_{pp'}}  \
\sqrt{D'_{\{\beta\}}}\, \e^{i\theta_{\{p'\}\{\beta\}}}\ ,
\end{equation}
which is of the required form.
\end{proof}

\noindent
{\bf Remark 2:} In the proof of the previous proposition nothing depended on having a finite number $m$ of modes 
in the first partition and a finite number $M-m$ of modes in the second partition. 
The result thus extends to the case of infinite disjoint sets of modes 
for all normalized pure states $\ket{\psi}$.
\hfill$\Box$
\medskip

Examples of $N$ fermions pure separable states are the Fock states; indeed, recalling
(\ref{2}), they can be recast in the form (\ref{8}):
\begin{equation}
| k_1, \ldots, k_m; \alpha_{m+1}, \ldots, \alpha_M\rangle=
\left[(\hat{a}_1^\dag)^{k_1}\cdots(\hat{a}_m^\dag)^{k_m}\right] \times
\left[(\hat{a}_{m+1}^\dag)^{\alpha_{m+1}}\cdots(\hat{a}_M^\dag)^{\alpha_M}\right]\ket{0}\ ,
\label{18}
\end{equation}
where $\cal P$ and $\cal Q$ are now monomials in the creation operators of the two partitions.
By varying $k_i, \alpha_j\in \{0,1\}$ and the integer $k=\sum_{i=1}^m k_i$, such that $0\leq k \leq N$,
these states generate the whole Hilbert space $\cal H$.
This basis states can be relabeled in a different, more convenient way as:
\begin{equation}
| k, \sigma; N-k, \sigma'\rangle\ ,\quad \sigma=1,2, \ldots, D_k\equiv{m\choose k}\ ,\
\sigma'=1, 2,\ldots, D'_{N-k}\equiv{M-m\choose N-k}\ ;
\label{19}
\end{equation}
the integer $k$ gives the number of fermions occupying the first $m$ modes, $k\leq m$,
while $\sigma$ counts the different ways in which those particles can fill those modes;
similarly, $\sigma'$ labels the ways in which the remaining $N-k$ fermions
can occupy the other $M-m$ modes.%
\footnote{In order to completely identify the basis states,
two extra labels $\sigma$ and $\sigma'$ are needed for each value of $k$, so that
these labels (and the range of values they take) are in general $k$-dependent: in order to keep
the notation as a simple as possible, in the following these dependences will be tacitly understood.}
In this new labelling, the property of orthonormality of the states
in (\ref{19}) simply becomes:\hfill\break
$\langle k, \sigma; N-k, \sigma'| l, \tau; N-l, \tau'\rangle=
\delta_{kl}\,\delta_{\sigma\tau}\,\delta_{\sigma'\tau'}$.

For fixed $k$, the basis vectors $\{| k, \sigma; N-k, \sigma'\rangle\}$ span a subspace ${\cal H}_k$
of dimension $D_k\, D'_{N-k}$; the union of all these orthogonal subspaces give the whole fermion
Hilbert space $\cal H$, recovering its dimension $D$ \cite{Prudnikov}:
\begin{equation}
\sum_{k=0}^N D_k\, D'_{N-k}=D= {M\choose N}\ .
\label{2-8}
\end{equation}

\noindent
{\bf Remark 3:} Note that the space ${\cal H}_k$ is naturally isomorphic to the tensor
product space $\mathbb{C}^{D_k}\otimes \mathbb{C}^{D'_{N-k}}$; through this isomorphism, the states
$| k, \sigma; N-k, \sigma'\rangle$ can then be identified with the
corresponding basis states of the form $| k, \sigma\rangle \otimes | N-k, \sigma'\rangle$.
This observation will be useful below in the classification of entangled fermion states.
\hfill$\Box$
\medskip

Using the above notation, a generic fermion mixed state $\rho$
can then be written as:
\begin{equation}
\rho=\sum_{k,l=N_-}^{ N_+ }\ \sum_{\sigma,\sigma',\tau,\tau'}\
\rho_{k \sigma\sigma', l\tau\tau'}\ | k, \sigma; N-k, \sigma'\rangle \langle l, \tau; N-l, \tau' |\ ,
\quad \sum_{k=N_-}^{ N_+ }\ \sum_{\sigma,\sigma'}\
\rho_{k \sigma\sigma', k\sigma\sigma'}=1\ .
\label{22}
\end{equation}
where $N_-={\rm max}\{0,N-M+m\}$ and $N_+={\rm min}\{N,m\}$ are the minimum and maximum
number of fermions that the first partition can contain, due to the exclusion principle.

The set of all states form a convex set, whose extremals are given by the pure ones;
{\sl Proposition 1} can then be used to characterize separable mixed states:

\begin{corollary}
A mixed state $\rho$ as in (\ref{22}) is
$(m, M-m)$-separable if and only if it is the
convex combination of projectors on pure $(m, M-m)$-separable
states; otherwise, the state $\rho$ is $(m, M-m)$-entangled.
\end{corollary}

In general, to determine whether a given density matrix $\rho$ can be written in separable 
form is a hard task and one is forced to rely on suitable separability tests, that however
are in general not exhaustive.

One of such tests has already been introduced through {\sl Corollary 1}, and is peculiar
to fermion systems: if there exists a local operator $A_1^o A_2^o\,$, with both
components odd under the action of the involution $\Theta$ ({\it cf.} {\sl Definition 4}),
such that the expectation value $\langle A_1^o A_2^o \rangle={\rm Tr}\big[\rho A_1^o A_2^o\big]$ 
is non vanishing, then the state $\rho$ is surely entangled. 
This entanglement criterion turns out to be exhaustive in the case of bipartitions of type $(1, M-1)$,
{\it i.e.} when the first partition contains just one mode, while the second the remaining $M-1$ ones.

\begin{proposition}
A $M$-mode state of an $N$-fermion system is entangled with respect to
the bipartition into one mode and the rest if and only if its expectation value
is nonvanishing on a local operator whose components both belong to the odd
part of the corresponding operator algebras.
\end {proposition}

\begin{proof}
A generic state of the system can be written in as
in (\ref{22}), though dropping the unnecessary primed greek labels:
\begin{equation}
\rho=\sum_{k,l=0}^1\ \sum_{\sigma,\tau}\
\rho_{k \sigma, l\tau}\ | k; N-k, \sigma\rangle \langle l; N-l, \tau |\ ,
\quad \sum_{k=0}^1\ \sum_{\sigma}\
\rho_{k \sigma, k\sigma}=1\ .
\label{23}
\end{equation}
It can be further decomposed into a ``diagonal'' and ``off-diagonal'' part
\begin{equation}
\rho=\rho_d+\eta\ ,
\label{24}
\end{equation}
with
\begin{eqnarray}
\label{13}
&&\rho_d=\sum_{k=0}^1\ \sum_{\sigma,\tau}\
\rho_{k \sigma, k\tau}\ | k; N-k, \sigma\rangle \langle k; N-k, \tau |\ ,\\
\label{14}
&&\eta=\sum_{k=0}^1\ \sum_{\sigma,\tau}\
\rho_{k \sigma, (1-k)\tau}\ | k; N-k, \sigma\rangle \langle 1-k; N+k-1, \tau |\ .
\end{eqnarray}
Clearly, only $\eta$ can give a nonvanishing contribution to the expectation value
${\rm Tr}\big[\rho A_1^o A_2^o\big]$, where $A_1^o A_2^o$ is a local operator
whose components are both odd.%
\footnote{Note that $\eta$ is not a state, since in general
it is not positive; it is the difference of two density matrices.}
 Therefore,
the {\sl Proposition} is proven once we show that the state $\rho_d$ can be written in separable form.
In order to prove this, let us make a change of basis in the second $M-1$ partition
passing from the Fock states to another set of separable states adapted to $\rho$,
such that its components along this new basis satisfy: 
$\rho_{k \sigma, k\tau}=\rho_{k \sigma, k\sigma}\ \delta_{\sigma\tau}$; notice that this is always possible
through suitable local, unitary transformations diagonalizing the two matrices
${\cal M}_{\sigma\tau}^{(k)}\equiv\big[\rho_{k \sigma, k\tau}\big]$, $k=0,1$. In this new basis,
$\rho_d$ results a convex sum of projections on separable pure states and therefore it is itself separable.
In conclusion, $\rho$ is entangled if and only if its ``off-diagonal'' part $\eta$ is nonvanishing.
\end{proof}

Another very useful entanglement criteria involves the operation
of partial transposition \cite{Peres,Horodecki2}: 
a state $\rho$ for which the partially transposed density matrix $\tilde\rho$ is
no longer positive is surely entangled. This lack of positivity can be quantified by the
so-called negativity \cite{Zyczkowski,Vidal1}:
\begin{equation}
{\cal N}(\rho)=\, {1\over2}\left( {\rm Tr}\Big[\sqrt{\tilde\rho^\dagger \tilde\rho}\Big]
- {\rm Tr}[\rho]\right)\ .
\label{27}
\end{equation}
which is nonvanishing only in presence of a non positive $\tilde\rho$.

In the case of the $(1, M-1)$-bipartition considered above, 
the partial transposition operation applied to the first partition gives
results that are completely equivalent to the
ones obtained in {\sl Proposition 2}. Indeed, explicit computation shows that
the negativity of the state in (\ref{23}) is nonvanishing if and only if
at least one of the off-diagonal components $\rho_{k \sigma, (1-k)\tau}$, $k=0, 1$, 
are nonzero.

Although not exhaustive, the partial transposition criterion results more powerful
than the one based on {\sl Corollary 1}, allowing a complete characterization
of the structure of entangled $N$-fermion states.

\begin{proposition}
A generic $(m, M-m)$-mode bipartite state (\ref{22}) is entangled
if and only if it can not be cast in the following block diagonal form
\begin{equation}
\rho=\sum_{k=N_-}^{ N_+ } p_k\ \rho_k\ ,\qquad 
\sum_{k=N_-}^{ N_+ } p_k=1\ ,\quad {\rm Tr}[\rho_k]=1\ ,
\label{28}
\end{equation}
with
\begin{equation}
\rho_k=\sum_{\sigma,\sigma',\tau,\tau'}\
\rho_{k \sigma\sigma', k\tau\tau'}\ | k, \sigma; N-k, \sigma'\rangle \langle k, \tau; N-k, \tau' |\ ,
\quad \sum_{\sigma,\sigma'}\rho_{k \sigma\sigma', k\sigma\sigma'}=1\ ,
\label{29}
\end{equation}
({\it i.e.} at least one of its non-diagonal coefficients $\rho_{k \sigma\sigma', l\tau\tau'}$, $k\neq l$,
is nonvanishing),
or, if it can, at least one of its diagonal blocks $\rho_k$ is non-separable.%
\footnote{For each block $\rho_k$, separability is understood with reference to
the isomorphic structure $\mathbb{C}^{D_k}\otimes \mathbb{C}^{D'_{N-k}}$ 
mentioned before (see {\sl Remark 3}).}
\end{proposition}

\noindent
{\sl Proof.} Assume first that the state $\rho$ can not be written in block diagonal form;
using techniques similar to the one adopted in dealing with boson systems \cite{Benatti3},
one can show that it is not left positive by the operation
of partial transposition and therefore it is entangled. Next, take $\rho$
in block diagonal form as in (\ref{28}), (\ref{29}) above. If all its blocks
$\rho_k$ are separable, then clearly $\rho$ itself results separable.
Then, assume that at least one of the diagonal blocks
is entangled. By mixing it with the remaining blocks as in (\ref{28}) will not
spoil its entanglement since all blocks $\rho_k$ have support on orthogonal spaces;
as a consequence, the state $\rho$ results itself non-separable.
\hfill$\Box$

\medskip\noindent
Using this result,
one can now study characteristic properties and geometry of the space $\cal F$ of $N$-fermion states.

The first question that naturally arises concerns the strength 
of the entanglement content of a fermion state $\rho$ 
against mixing with other states.
This is measured by the so-called
``robustness of entanglement'' $R(\rho)$ \cite{Vidal,Steiner,Vidal1,Plenio}:
it is the smallest, non-negative value the parameter $t$ can take so
that the (un-normalized) combination $\rho + t\rho_{\rm sep}$ is separable,
where $\rho_{\rm sep}$ span all separable fermion states.

\begin{proposition}
The robustness of entanglement of a generic $(m, M-m)$-mode
bipartite $N$-fermion state $\rho$ is given by
\begin{equation}
R(\rho)=\sum_{k=0}^N p_k\, R(\rho_k)\ ,
\label{3-2}
\end{equation}
for states that are in block diagonal form as in (\ref{28}), (\ref{29}), while it is infinitely large
otherwise.
\end{proposition}

\noindent
The proof is very similar to the one given in \cite{Benatti4} for bosonic states, so that
it will not be repeated here. Nevertheless, it is worth stressing that, as in the case of bosons, 
fermion entanglement is in general much more robust than the one found 
in systems of distinguishable particles.
Indeed, recalling the previous
{\it Proposition 3}, one has that separable $N$-fermion states
must necessarily be in block diagonal form. If the state $\rho$ is not in this form, it can
never be made block diagonal by mixing it with any separable one; 
therefore, in this case, the combination 
$\rho + t\, \rho_{\rm sep}$ will never be separable, unless $t$ is infinitely large,
giving infinite robustness to the entangled state $\rho$.

A similar argument allows to conclude that the structure of the space $\cal S$ of
separable fermion states is rather special: there always exist small perturbations 
of separable, necessarily block diagonal, states that
make them not block diagonal, hence entangled.
This result should be compared to the one valid
in the case of distinguishable
particles, where instead almost all separable states remain separable under sufficiently small
arbitrary perturbations \cite{Bengtsson,Bandyopadhyay}.

Among the separable $N$-fermion states, the totally mixed one,
\begin{equation}
\rho_{\rm mix}={1\over D}\sum_{k=N_-}^{N_+}\ \sum_{\sigma,\sigma'}\
| k, \sigma; N-k, \sigma'\rangle \langle k, \sigma; N-k, \sigma' |\ ,
\label{4-7}
\end{equation}
stands out because of its special properties. First of all,
it lays on the border of the space $\cal S$ of separable states, since
in its vicinity one can always find non-separable states of the
form $\rho_{\rm mix}+\epsilon\, \rho_{\rm ent}$, $\epsilon >0$,
with $\rho_{\rm ent}$ any entangled state not in block-diagonal form.%
\footnote{On the contrary, recall that in the case
of distinguishable particles, $\rho_{\rm mix}$ always lays in the interior of $\cal S$
\cite{Zyczkowski,Bengtsson}.}
Further, $\rho_{\rm mix}$ is the only state that remains separable
for any choice of bipartition. Indeed, for any state $\rho_{\rm sep}\neq \rho_{\rm mix}$, separable
in a given $M$-mode bipartition, it is possible to find a
unitary Bogolubov transformation, defining a new $M$-mode bipartition, 
that maps it into an entangled one \cite{Benatti4}. Only the state proportional to the unit matrix
remains unchanged under any unitary transformation.

These results allows a rather precise description of the geometrical structure
of the space $\cal F$ of $N$-fermion states. As discussed above,
by fixing a bipartition one selects the set $\cal S$ of separable states,
which forms a subspace of the convex space $\cal F$.
Changing the bipartition through a Bogolubov transformation produces a new separable
subspace, having in general only one point in common with the starting one,
$\rho_{\rm mix}$. Therefore, the global geometrical structure of the state space $\cal F$
presents a sort of star-like shape formed by the various separable bipartition 
subspaces, all sharing just one point,
the totally mixed state.

As a final comment, notice that all above results can be generalized to the case
of systems where the total number of particles is not fixed,
but commutes with all physical observables.%
\footnote{In other terms, we are in
presence of a superselection rule \cite{Bartlett}. For a similar reason, {\it i.e.} the
conservation of the fermion ``charge'',
and in contrast with the boson case, fermion systems with fluctuating number of particles
result unphysical \cite{Wick,Moriya}, so that density matrices that are coherent mixtures of 
states with different $N$ are not admissible.}
In such a situation, a general density matrix $\rho$ can be written as an incoherent mixture 
of states $\rho_N$ with fixed number $N$ of
fermions:
\begin{equation}
\rho=\sum_N \lambda_N \rho_N \ ,\qquad \lambda_N\geq 0\ ,\qquad \sum_N \lambda_N=1\ .
\label{5-1}
\end{equation}
The state $\rho$ is a convex combination of matrices $\rho_N$ having support
on orthogonal spaces, and therefore all previous arguments and results
hold true for each component $\rho_N$.

\section{Applications to quantum metrology}

One of the most promising developments in quantum technology, {\it i.e.} the application
of quantum physics to practical technological realizations, is the possibility of achieving
measurements of physical parameters with unprecedented accuracy.%
\footnote{The literature on the subject is vast; for a partial list, 
see \cite{Caves1}-\cite{Giovannetti2} and references therein.}
In a generic detection
scheme, the parameter $\theta$ to be measured, typically a phase, is encoded into
a state transformation occurring inside a measurement apparatus, generally
an interferometric device. In the most common case of linear setups, this transformation
can be modelled by a unitary mapping, $\rho \to \rho_\theta$, sending the
initial state $\rho$ into the final parameter-dependent outcome state:
\begin{equation}
\rho_\theta= e^{i\theta J}\, \rho\, e^{-i\theta J}\ ,
\label{33}
\end{equation}
where $J$ is the devices-dependent, $\theta$-independent operator generating
the state transformation.
The task of quantum metrology is to determine the ultimate bounds on the accuracy with which
the parameter $\theta$ can be obtained through a measurement of $\rho_\theta$
and to study how these bounds scale with the available resources.

General quantum estimation theory allows a precise determination of
the accuracy $\delta\theta$ with which the phase $\theta$ can be obtained in a measurement involving
the operator $J$ and the initial state $\rho$; one finds that $\delta\theta$ is limited by the
following inequality \cite{Helstrom}-\cite{Paris}:
\begin{equation}
\delta\theta\geq {1\over \sqrt{F[\rho, J]}}\ ,
\label{34}
\end{equation}
where the quantity $F[\rho, J]$ is the so-called ``quantum Fisher information''.
It can be defined through the symmetric logarithmic derivative $L$, 
$\partial_\theta\rho_\theta\big\vert_{\theta=0}=(\rho\,L\,+\,L\,\rho)/2=-i\,[J\,,\,\rho]$,
as
\begin{equation}
\label{35}
F[\rho,J]:={\rm tr}\big[\rho\,L^2\big]\ .
\end{equation}
Given a spectral decomposition of the input state, $\rho=\sum_j r_j\,\vert r_j\rangle\langle r_j\vert$, 
one explicitly finds
\begin{equation}
\label{36}
F[\rho,J]=2\,\sum_{i,j\,;\,r_i+r_j\neq0}\frac{(r_i-r_j)^2}{r_i+r_j}\,
\Big|\langle r_i\vert J\,\vert r_j\rangle\Big|^2\ ,
\end{equation}
which explicitly shows that $F[\rho,J]$ is independent from the parameter $\theta$ to be estimated.
Further, the quantum Fisher information is a continuous, convex function
of the state $\rho$, and in general satisfies the inequality \cite{Luo,Braunstein}
\begin{equation}
F[\rho, J]\leq 4\, \Delta^2_{\rho} J\ ,
\label{37}
\end{equation}
where $\Delta^2_{\rho} J\equiv\big[ \langle J^2\,\rangle-\langle J\,\rangle^2\big]$ 
is the variance of the operator $J$ in the state $\rho$, the equality holding
only for pure initial states.

As a consequence of (\ref{34}), better resolution in $\theta$-estimation 
corresponds to a larger quantum Fisher information. Therefore,
once the measuring apparatus is given, {\it i.e.} the operator $J$ is fixed, one can optimize
the precision with which $\theta$ is determined by choosing an initial state $\rho$
that maximizes $F[\rho, J]$.

In the case of devices using a system of $N$ {\sl distinguishable} particles,
it has been shown that for any separable state $\rho_{\rm sep}$
the quantum Fisher information is bounded by $N$ \cite{Smerzi3}:
\begin{equation}
F\big[\rho_{\rm sep}, J\big]\leq N\ .
\label{26}
\end{equation}
This means that by feeding the measuring apparatus with separable initial states, the best 
achievable precision in the determination of the phase shift $\theta$ is bounded in this case
by the so-called shot-noise limit:
\begin{equation}
\delta\theta\geq{1\over\sqrt{N}}\ .
\label{27}
\end{equation}
This is also the best result
attainable using classical ({\it i.e.} non quantum) devices:
the accuracy in the estimation of $\theta$ scales at most with the inverse
square root of the number of available resources.
Instead, quantum equipped metrology allows to go below the shot-noise limit and 
in principle construct a new generation of sensors reaching unprecedented precision.
And indeed, various detection protocols and input states $\rho$ have been proposed,
all able to yield sub shot-noise sensitivities. 
Notice that, in view of the inequality (\ref{26}), these input states need to be entangled.

This conclusion holds when the metrological devices used to estimate the physical
parameter $\theta$ are based on systems of distinguishable particles. When dealing with
{\sl identical} particles, the above statement is not strictly correct and needs to be rephrased.
Indeed, in the case of bosonic systems it has been explicitly shown that sub shot-noise
sensitivities may be obtained also via a non-local
operation acting on separable input states \cite{Benatti1}. 
In other terms, although some sort of non-locality is needed in order to go below the
shot-noise limit, this can be provided by the measuring apparatus itself and
not by the input state $\rho$, that indeed can be separable. This result
has clearly direct experimental relevance, since the preparation of
suitable entangled input state may require in practice a large amount of resources.

When dealing with systems of $N$ fermions, the situation may appear more problematic, since,
due to the exclusion principle, a single mode can accommodate at most one fermion;
therefore, the scaling with $N$ of the sensitivity in the estimation 
of the parameter $\theta$ may worsen when compared to the boson case employing similar resources.
In the case of bosons, a two-mode apparatus, {\it e.g.} a double-well interferometer, 
filled with $N$ particles is sufficient to reach sub shot-noise sensitivities.
Instead, with fermions, a multimode interferometer \cite{Dariano1}-\cite{Cooper} is needed in order to reach
comparable sensitivities.

As an example, consider a system of $N$ fermions in $M$ modes, with $M$ even,
and let us fix the balanced bipartition $(M/2, M/2)$, in which each of the two parts
contain $m=M/2$ modes, taking for simplicity $N\leq m$. As generator
of the unitary transformation $\rho\to\rho_\theta$ inside the measuring apparatus
let us take the following operator:
\begin{equation}
J_x^{(1)}=\frac{1}{2}\sum_{k=1}^m \omega_k\ \Big(a_k^\dagger a_{m+k} + a_{m+k}^\dagger a_k\Big)\ ,
\label{40}
\end{equation}
where $\omega_k$ is a given spectral function, {\it e.g.} $\omega_k \simeq k^p$, with $p$ integer.
The apparatus implementing the above state transformation is clearly non-local
with respect to the chosen bipartition: $e^{i\theta J_x^{(1)}}$ can not be written
as the product $A_1 A_2$ of two components made of operators referring only
to the first, second partition, respectively. It represents a generalized,
multimode beam splitter, and the whole measuring device behaves as
a multimode interferometer. 

Let us feed the interferometer with a pure initial state, $\rho=|\psi\rangle\langle \psi|$,
\begin{equation}
|\psi\rangle=| \underbrace{1,\dots,1}_N,\underbrace{0,\ldots,0}_{m-N}\, ;\ \underbrace{0,\ldots,0}_m\, \rangle
=a_1^\dagger a_2^\dagger\cdots a_N^\dagger |0\rangle\ ,
\label{41}
\end{equation}
where the fermions occupying the first $N$ modes of the first partition; $|\psi\rangle$ is a Fock state
and therefore it is separable, as discussed in the previous section. The quantum Fisher information
can be easily computed since it is now proportional to the variance of $J_x^{(1)}$:
\begin{equation}
F\big[\rho, J_x^{(1)}\big]= 4\Delta^2_{\rho} J_x^{(1)}=\sum_{k=1}^N \omega_k^2\ .
\label{42}
\end{equation}
Unless $\omega_k$ is $k$-independent, $F\big[\rho, J_x^{(1)}\big]$ is larger than $N$ and therefore
the interferometric apparatus can beat the shot-noise limit in $\theta$-estimation, even starting
with a separable state.
Actually, for $\omega_k \simeq k^p$, one gets: $F\big[\rho, J_x^{(1)}\big]\simeq {\cal O}(N^{2p+1})$.

\bigskip
\noindent
{\bf Remark 4:} {\sl i)} Notice that, in this case, it is not the entanglement
of the initial state that help overcoming the shot-noise-limit
in the phase estimation accuracy; rather, it is the non-local character of the
rotations operated by the apparatus on an initially separable state
that allows $\delta\theta$ to be smaller than $1/\sqrt{N}$.

\noindent
{\sl ii)} In the case of systems made of $N$ {\sl distinguishable} particles, 
for a collective operator $J=\sum_{i=1}^N J^{(i)}$, where $J^{(i)}$ are single particle $SU(2)$ rotation generators,
and any state $\rho$,
the following general bound on the quantum Fisher information holds:
\begin{equation}
F\big[\rho, J\big]\leq N^2\ ,
\label{43}
\end{equation}
providing an absolute limit for the accuracy on the parameter estimation: $\delta\theta\geq 1/N$. 
When the equality holds, one reaches the so-called Heisenberg limit, the ultimate sensitivity
allowed by quantum metrology in this case. Instead, in the scenario described above, one can reach
sub-Heisenberg sensitivities. This possibility has been discussed before, using however
non-linear metrology \cite{Luis}-\cite{Hall}, {\it i.e.} for interferometric apparata that can not be described
in terms of single particle rotations. Further, note that the result (\ref{42})
and the ability to go beyond the Heisenberg limit is not a ``geometrical'' phenomena attributable to
a phase accumulation even on empty modes \cite{Dariano1}; rather, it is a genuine quantum effect,
that scales as a function of the number of fermions, the resource available in the measure.

\noindent
{\sl iii)} When dealing with boson systems, an interferometer based on standard beam splitters 
suffices in order to reach sub shot-noise
accuracy in parameter estimation. In such case the relevant operator is $J_x$, obtained from (\ref{40})
by removing the spectral function $\omega_k$; it belongs to an $su(2)$ algebra.
Similarly, also the generalized beam splitter operator (\ref{40}) is part of a Lie algebra, although
infinite dimensional. Indeed, let us define the three operators:
\begin{eqnarray}
\label{44}
&&J_x^{(n)}=\frac{1}{2}\sum_{k=1}^m (\omega_k)^n\ \Big(a_k^\dagger a_{m+k} + a_{m+k}^\dagger a_k\Big)\ ,\\
\label{45}
&&J_y^{(n)}=\frac{1}{2i}\sum_{k=1}^m (\omega_k)^n\ \Big(a_k^\dagger a_{m+k} - a_{m+k}^\dagger a_k\Big)\ ,\\
\label{46}
&&J_z^{(n)}=\frac{1}{2}\sum_{k=1}^m (\omega_k)^n\ \Big(a_k^\dagger a_k - a_{m+k}^\dagger a_{k+m}\Big)\ .
\end{eqnarray}
One easily checks that they satisfies the following commutation relations:
\begin{equation}
\Big[J_i^{(n)}, J_j^{(m)}\Big]=i\varepsilon_{ijk}J_k^{(n+m)}\ ,\qquad i,j,k=x,y,z\ ,\qquad n,m\in\mathbb{Z}\ ,
\label{47}
\end{equation}
defining the $su(2)$ loop-algebra ({\it i.e.} a centerless Kac-Moody algebra) \cite{Fuchs,Francesco}.
\hfill$\Box$

\medskip
In general, some sort of quantum non-locality is nevertheless needed in order attain sub shot-noise accuracy
in phase estimation. This can be most simply appreciated by changing the $M$-mode bipartition of our system of
$N$ fermions. Let us introduce new creation and annihilation operators 
$b_k^\dagger$, $b_k$ through the following Bogolubov transformations:
\begin{equation}
b_k = \frac{a_k + a_{m+k}}{\sqrt2}\ ,\qquad b_{m+k} = \frac{a_k - a_{m+k}}{\sqrt2}\ ,\qquad k=1,2,\ldots, m\ ,
\label{}
\end{equation}
together with the hermitian conjugate ones; 
the new $M$ modes still obey canonical anticommutation relations: $\{b_k,\, b_l^\dagger\}=\delta_{kl}$.
In this new representation, the operator $J_x^{(1)}$ in (\ref{40}) get transformed into $J_z^{(1)}$,
\begin{equation}
J_x^{(1)}\to J_z^{(1)}=\frac{1}{2}\sum_{k=1}^m \omega_k\ \Big(b_k^\dagger b_k - b_{m+k}^\dagger b_{m+k}\Big)\ .
\label{}
\end{equation}
Therefore, choosing again a balanced bipartition, $(M/2,\, M/2)$, in which half of the $b_k$ modes are 
in the first component, while the remaining half in the second one, the unitary transformation
$\rho\to \rho_\theta$ inside the apparatus is now represented by a local operator
\begin{equation}
e^{i\theta\, J_z^{(1)}} = e^{i\theta\, \sum_{k} \omega_k\, b_k^\dagger b_k/2}\ 
e^{-i\theta\, \sum_{k} \omega_k\, b_{m+k}^\dagger b_{m+k}/2}\ .
\label{}
\end{equation}
However, the initial state $|\psi\rangle$ is no longer separable in the new bipartition,
since, in the new language, it results a linear combination of $2^N$ different Fock states:
\begin{equation}
|\psi\rangle = \left(\frac{b_1^\dagger +b_{m+1}^\dagger}{\sqrt2}\right) 
\left(\frac{b_2^\dagger +b_{m+2}^\dagger}{\sqrt2}\right)\cdots 
\left(\frac{b_N^\dagger +b_{N+1}^\dagger}{\sqrt2}\right) |0\rangle\ .
\label{}
\end{equation}
Despite these changes, the value of the quantum Fisher information for the given initial state
and observable is unchanged and still expressed by (\ref{42}),
since it does not depend on the representation used to compute it. Therefore,
if one is able to build an interferometric setup that
can be described in terms of the modes $b_k^\dag$, $b_k$ instead of the original
modes $a_k^\dag$, $a_k$, then the accuracy $\delta\theta$ with which the phase $\theta$
may be determined can still be better than the shot-noise limit. In such a case, the improvement
in sensitivity is due to the
entanglement of the initial state and not to the non-locality of the transformation 
that takes place inside the apparatus.

As a further remark, notice that in practical applications,
instead of using a generalized rotation through operators of the form (\ref{44})-(\ref{46}),
it might be more convenient to implement parameter estimation via the dynamical state
transformation generated by an hamiltonian operator. A generic quadratic hamiltonian for our $N$ fermions system 
can be written in the form
\begin{equation}
H= \sum_{k=1}^M \Omega_k\, a_k^\dagger\, a_k \ ,
\label{}
\end{equation}
with $\Omega_k$ a given energy dispersion relation. In most situations, the dependence on the
parameter $\theta$ to be estimated arises as a proportionality coupling constant
multiplying the hamiltonian, $H_\theta\equiv\theta H$; in this case,
the finite-time dynamical transformation occurring in the system
is described by $e^{-i t H_\theta}$.
This operator is clearly 
local in any $(m, M-m)$ bipartition, since it is the product of $M$ transformations
in the various modes: $e^{-i t H_\theta}=\prod_k e^{-i t\,\theta\, \Omega_k\, a_k^\dagger a_k}$.
For an entangled initial state of the form%
\footnote{This state is the multimode, $N$-fermion generalization of the
state (\ref{6}) considered in Section 2.}
\begin{equation}
|\Phi\rangle=\frac{1}{\sqrt2}\Big(|N;0\rangle+|0;N\rangle\Big)\ ,
\label{}
\end{equation}
where the two states $|N;0\rangle$, $|0;N\rangle$ represents the situation in which the $N$ fermions are
all in the first, second component, respectively, of a generic $(m, M-m)$ bipartition,
the quantum Fisher information can be readily computed:
\begin{equation}
F\Big[|\Phi\rangle\langle\Phi|, H\Big]=\left( \sum_{k=1}^N\big(\Omega_{m+k}-\Omega_k\big)\right)^2\ .
\end{equation}
In the case of a linear dispersion relation, $\Omega_k\sim k$, it reduces to the simple form
$F\big[|\Phi\rangle\langle\Phi|, H\big]=m^2\, N^2$, providing a sub Heisenberg-like sensitivity in
the estimation of the quantity $t \theta$, hence of the parameter $\theta$, once the evolution
time $t$ is fixed. Notice that in $F$ the factor $N^2$ is a genuine quantum effect, while the dependence
on the number of modes is a ``geometrical'' effect due to phase-accumulation among all,
even empty, modes ({\it cf.} {\sl Remark 4, ii)}).

\eject

\section{Outlook}

One of the most important tasks in modern quantum physics is the characterization and quantification
of non-classical correlations, as they allow on one hand the implementation of classically
unavailable protocols in information theory, on the other hand the realization of quantum devices
and sensors outperforming the present available ones. In these developments, 
thank to the recent advances in quantum optics, ultracold and superconducting physics,
many-body systems composed by identical particles are playing a prominent role.

For such systems though, the usually adopted definitions of separability and quantum correlation
are no longer viable since, due to the indistinguishability of the microscopic constituents,
the natural particle Hilbert space decomposition on which these concepts are based
is lost. One should then resort to a more general definition of locality,
no longer given {\it a priori} once for all, rather, based on a choice of a
bipartition (or more in general multipartition) of the operator algebra of ``observables''
of the system. In this framework, a state is declared separable if its expectation value
on all local operators can be written in a product form, {\it i.e.} as a convex combination of products
of local expectations. This new approach to quantum non-locality is valid in all situations
and, in particular, it reduces to the standard one when applied to systems of distinguishable particles.

The physical, characteristic properties of this new, generalized definition of separability,
previously analyzed in a bosonic setting,
has been studied here in the case of fermion systems. We have focused on many-body systems composed of a fixed
number $N$ of fermions that can occupy a given set of different states or modes. 
We stress that this model represents
a very general paradigm, able to describe the behaviour of various different situations
in atomic and condensed matter physics, as those occurring in 
quantum phase transitions and matter interference phenomena.

The treatment of fermion systems require more care than in the boson case because
of the anticommutative character of the basic fermion algebra. As a result, in contrast
to the bosonic case, the notion
of locality for fermion systems is not directly related with that of commutativity;
nevertheless, the intuition that entanglement should
be connected to the presence of non-classical correlations revealed through averages of local operators
turns out to be correct also in this case. As a byproduct, a new entanglement criterion for fermion states 
is obtained. Using this criterion together with the partial transposition one, a complete
classification of entangled $N$-fermion states have been explicitly given.

Similarly to what happens with $N$-boson states, the entanglement contained in $N$-fermion states turns out
to be much more stable against mixing with other states than the one found in systems of distinguishable
particles; this makes many-body systems made of identical constituents even more attractive for use
in quantum technology applications.

In this respect, quantum metrology is the natural context in which
systems of $N$-fermions can be employed to construct quantum devices
that outperform classical ones. Indeed, as discussed in the last section,
multimode fermion quantum interferometers can be used to improve the accuracy in
parameter estimation much beyond the so-called, classical, shot-noise limit,
provided some sort of quantum non-locality is present in the measuring process.
However, this required non-locality 
need not be encoded in the initial $N$-fermion state:
it can be provided by the interferometric apparatus itself, which at this point
can be fed with an initial separable state. As a result, no preliminary, resource consuming,
entanglement operation (like ``squeezing'') on the state entering the apparatus is need in order to 
get sub shot-noise accuracies in parameter estimation. This fact clearly opens new perspectives
in the realization of many-body based quantum sensors
capable of outperform any available apparatus dedicated to
the measurement of ultraweak effects.

Finally, let us briefly consider the case in which the fermionic system
is describable in terms of a set of Majorana, hermitian operators $c_i$,
$i=1,2,\ldots, 2M$, obeying the algebraic relations: $\{ c_i, c_j \}=2\delta_{ij}$.
Clearly, also in this case the set of all polynomials in the operators $c_i$
form an algebra $\cal C$, to which the system observables belong.
The adopted notions of algebraic bipartition, locality and
separability (see {\sl Definition 1, 2, 3}) are very general and can be
applied also to $\cal C$, so that most of the general results obtained
in the case of complex fermion algebras hold also for the hermitian ones.
Nevertheless, the system Hilbert spaces differ in the two cases;
in particular, the Majorana algebra does not admit a Fock representation.
As a consequence, the detailed structure of entangled Majorana states 
differs from that reported in Sect.3, deserving a separate, expanded discussion
that will be reported elsewhere.

\end{document}